\documentclass[12pt]{article}
\usepackage[margin=1in]{geometry} 
\usepackage{amsmath,amsthm,amssymb,scrextend}
\usepackage{fancyhdr}
\usepackage{comment}
\usepackage{graphicx}
\usepackage{color}
\usepackage{cite}
\usepackage{authblk}

\usepackage{bbm}

\newtheorem{theorem}{Theorem}
\newtheorem{lemma}{Lemma}

\newcommand{\ratio}[0]{\alpha}

\begin{document}
 
\title{A Note on Approximate Revenue Maximization with Two Items}
\author{Ron Kupfer \thanks{School of Computer Science and Engineering, Hebrew University of Jerusalem}}
\date{}
\maketitle

\begin{abstract}
We consider the problem of maximizing revenue when selling $2$ items to a single buyer with known valuation distributions.
In \cite{hart2012approximate}, Hart and Nisan showed that selling each item separately using the optimal Myerson's price, gains at least half of the revenue attainable by optimal auction for two items.
We show that in case the items have different revenues when sold separately the bound can be tightened.
\end{abstract}

\section{Introduction}

Following Hart and Nisan work in \cite{hart2012approximate} we are looking at the scenario of a single seller selling multiple items to a single buyer. This paper is limited to the case of two items with independently distributed values and an additive valuation. The distributions (given by a cumulative distribution $F$) are known to the seller but not the chosen valuations.

For one item the problem is completely solved by Myerson's classic result \cite{myerson1981optimal}. If the seller offer to sell it for a price $p$ then the probability that the buyer will buy is $1 - F(p)$, and the revenue will be $p\cdot(1 - F(p))$. The seller will choose a price $p^*$ that maximizes this expression. Myerson's characterization of optimal auctions concludes that the take-it-or-leave-it offer at the above price $p^*$ yields the optimal revenue among all mechanisms. Myerson's result also applies when there are multiple buyers, in which case $p^*$ would be the reserve price in a second price auction.

Hart and Nisan showed that selling each item separately using the optimal Myerson's price, gains at least half of the revenue attainable by optimal auction for two items.

In addition, they showed that when the items valuations are drown from two i.i.d distribution the ratio between the revenues is tighter and stand on at least $73\%$.

We show that on the other hand, when the two distributions of the values correspond to two different revenues (when selling each item separately) the revenue ratio is again grater then half. Thus, covering another set of distributions.

\subsection*{Notations and Preliminaries}
\subsubsection*{Mechanisms}
\begin{itemize}
\item $x=(x_1,...,x_k)$ - buyer valuation getting each item.
\item $q(x)=(q_1,...,q_k)$ - denotes the probability
that item i is allocated to the buyer when his valuation is $x$
\item $s(x)=(s_1,...,s_k)$ - denotes the expected payment that the seller receives from the buyer when the buyer’s valuation is $x$
\item $b(x)=(b_1,...,b_k)$ - denotes the utility of the buyer when his valuation is $x$,
i.e., $b(x) = x\cdot q(x) - s(x)$
\item IR - Individually Rational - $b(x)\geq 0$ for all $x$
\item IC - Incentive Compatible - for all $x,x'$: $x\cdot q(x) - s(x)\geq x\cdot q(x') - s(x')$
\end{itemize}

\subsubsection*{Revenue}

For a cumulative distribution $F$ on \(\mathbb{R}^{k}_{+}\) (for \(k \geq 1\)), we consider the optimal revenue
obtainable from selling $k$ items to a (single, additive) buyer whose valuation for the $k$
items is jointly distributed according to $F$:
\begin{itemize}
\item \(Rev(F)\) is the maximal revenue obtainable by any incentive compatible
and individually rational mechanism.
\item \(SRev(F)\) is the maximal revenue obtainable by selling each item separately.
\end{itemize}

\section{Result}
\begin{theorem}For every one-dimensional distributions $F_1$ and $F_2$,
$$\left(2-\frac{\alpha - 1 - ln(\alpha)}{1+\alpha}\right)SRev(F_1\times F_2) \geq Rev(F_1 \times F_2)$$
where $\alpha=\max{\{\frac{Rev(F_1)}{Rev(F_2)},\frac{Rev(F_2)}{Rev(F_1)}\}}$
\end{theorem}

From now on we denote $R_i=Rev(F_i)$. w.l.o.g. we assume that $R_1\geq R_2$ and denote $\frac{R_1}{R_2}$ as $\ratio$ ($\geq 1$).

\begin{lemma}\label{r1_r2_min_is_grater_then_rev} Let $A,\ B$ be two items with value distributions $X$ and $Y$ associated with $SRev$ of $R_1,\ R_2$ respectively. $R_1 + R_2 + \mathbb{E}(\min\{X,Y\})\geq \mathbb{E}s(X,Y) $
\end{lemma}
\begin{proof}
we use a similar technique to the one used in appendix A in \cite{hart2012approximate}. Take any IC and IR mechanism $(q,s)$. We will split its expected revenue into two parts, according to which one of $X$ and $Y$ is maximal: $\mathbb{E}(s(X,Y)) \leq \mathbb{E}(s(X,Y)\mathbbm{1}_{X\geq Y} ) + \mathbb{E}(s(X,Y )\mathbbm{1}_{Y\geq X})$

for every fixed value $y$ of $Y$ define a mechanism $(\tilde{q}^y, \tilde{s}^y)$ for $X$ by $\tilde{q}^y(x) := q_1(x,y)$ and $\tilde{s}^y(x) :=s(x,y)-yq_2(x,y)$ for every $x$ (so the buyer’s payoff remains the same: $\tilde{b}^y(x) = b(x,y)$).

The mechanism $(\tilde{q}^y, \tilde{s}^y)$ is IC and IR for $X$, since $(q,s)$ was IC and IR for $(X,Y)$ (only the IC constraints with $y$ fixed, i.e., $(x′,y)$ vs. $(x,y)$, matter). We do the same for $(\tilde{q}^x, \tilde{s}^x)$.
Therefore,
\begin{equation} \label{item1}
R_1 \geq \mathbb{E}_{x\sim X}(S_1^y(x)) \geq \mathbb{E}(S_1^y(x)\mathbbm{1}_{x\geq y}) \geq \mathbb{E}((s(x,y)-y)\mathbbm{1}_{X\geq y})
\end{equation}
\begin{equation} \label{item2}
R_2 \geq \mathbb{E}_{y\sim Y}(S_2^x(y)) \geq \mathbb{E}(S_2^x(y)\mathbbm{1}_{y\geq x}) \geq \mathbb{E}((s(x,y)-x)\mathbbm{1}_{Y\geq x})
\end{equation}

summing equations \ref{item1}, \ref{item2} and taking expectation over $X, Y$ we get: 
\begin{equation}
R_1 + R_2 \geq \mathbb{E}s(X,Y) - \mathbb{E}(Y\mathbbm{1}_{X\geq Y} + X\mathbbm{1}_{Y\geq X}) =  \mathbb{E}s(X,Y) - \mathbb{E}(\min\{X,Y\})
\end{equation}

\end{proof}

\begin{lemma}\label{r1_much_bigger_r2}
$\mathbb{E}(\min\{X,Y\})\leq (2+\ln(\ratio))R_2$
\end{lemma}
\begin{proof}

\[\mathbb{E}(\min\{X,Y\}) = \int_{0}^{\infty}{\bar{F_1}(u)\bar{F_2}(u)du}\]
\[ \leq \int_{0}^{R_2}{1\cdot du} + \int_{R_2}^{R_1}{\frac{R_2}{u} du} + \int_{R_1}^{\infty}{\frac{R_1R_2}{u^2} du}\ = R_2 + R_2\ln\left(\frac{R_1}{R_2}\right) + \frac{R_1R_2}{R_1} = R_2(2+\ln\left(\frac{R_1}{R_2}\right))\]
\end{proof}

combining the lemmas we get that
\[\mathbb{E}s(X,Y)\leq R_1 + (3+\ln(\frac{R_1}{R_2}))R_2\]

\[ = \left(2-\frac{\alpha - 1 - ln(\alpha)}{1+\alpha}\right)(R_1+R_2)\]

thus, finishing the proof of Theorem 1.

\bibliography{main.bib} 

\begin{thebibliography}{1}

\bibitem{hart2012approximate}
Sergiu Hart and Noam Nisan.
\newblock Approximate revenue maximization with multiple items.
\newblock {\em arXiv preprint arXiv:1204.1846}, 2012.

\bibitem{myerson1981optimal}
Roger~B Myerson.
\newblock Optimal auction design.
\newblock {\em Mathematics of operations research}, 6(1):58--73, 1981.

\end{thebibliography}
\bibliographystyle{plain}
\end{document}